\documentclass[reqno]{amsart}
\usepackage{amsmath}
\usepackage{amssymb}
\usepackage{amsthm}
\usepackage{amsfonts}
\usepackage{mathtools}
\usepackage{bm}
\usepackage{ulem}
\usepackage{fancyhdr}
\usepackage{hyperref}
\usepackage{framed}
\usepackage{float}
\usepackage[pdftex,final]{graphicx}
\usepackage{tikz}
\usetikzlibrary{calc,arrows}

\usepackage[iso-8859-7]{inputenc}

\theoremstyle{plain}
\newtheorem{thm}{Theorem}[section]

\newtheorem{prop}[thm]{Proposition}
\newtheorem*{theorem*}{Theorem}
\theoremstyle{definition}

\newtheorem{dfn}[thm]{Definition}

\numberwithin{equation}{section}

\newcommand{\mytilde}{\raise.17ex\hbox{$\scriptstyle\mathtt{\sim}$}}

\newcommand{\I}{\ensuremath{\mathcal{I}}}
\renewcommand{\S}{\ensuremath{\mathcal{S}}}
\newcommand{\Ni}{\ensuremath{|\mathcal{N}_i|}}
\newcommand{\RgeO}{\ensuremath{\mathbb{R}_{ \geq 0}}}
\newcommand{\Rat}[1]{\ensuremath{\mathbb{R}^{#1}}}

\begin{document}

\author{D. Boskos}
\address{Department of Automatic Control, School of Electrical Engineering, KTH Royal Institute of Technology, Osquldas v\"ag 10, 10044, Stockholm, Sweden}
\email{boskos@kth.se}

\author{D. V. Dimarogonas}
\address{Department of Automatic Control, School of Electrical Engineering, KTH Royal Institute of Technology, Osquldas v\"ag 10, 10044, Stockholm, Sweden}
\email{dimos@kth.se}

\begin{abstract}
The purpose of this report is to define abstractions for multi-agent systems under coupled constraints. In the proposed decentralized framework,  we specify a finite or countable transition system for each agent which only takes into account the discrete positions of its neighbors. The dynamics of the considered systems consist of two components. An appropriate feedback law which guarantees that certain performance requirements (eg. connectivity) are preserved and induces the coupled constraints and additional free inputs which we exploit in order to accomplish high level tasks. In this work we provide sufficient conditions on the space and time discretization of the system which ensure that we can extract a well posed and hence meaningful finite transition system.
\end{abstract}

\title{Decentralized Abstractions for Feedback Interconnected Multi-Agent Systems}
\maketitle

\section{Introduction}

Cooperative task planing under temporal logic specifications constitutes a highly active area of research which lies in the interface between computer science and modern control theory. One main challenge in this new interdisciplinary direction is the problem of defining appropriate abstractions for continuous time control systems and hence enabling the analysis and control of large scale systems or the achievement of high level plans. Robot motion planing and control constitutes a central field where this line of work is applied. In particular the use of a suitable discrete system's model allows the specification of high level plans, which under an appropriate equivalence notion between the continuous system and its discrete analog, can be converted to low level primitives such as feedback controllers, that are able to implement the high level tasks. Such tasks in the case of multiple mobile robots in an industrial workspace could include for example the following scenario. Robot 1 should periodically go from region $A$ to region $B$, while avoiding $C$ and after collecting an item of type $X$ from robot 2 at location $D$ and storing it at location $E$.

In order to accomplish high level plans, we need to specify a finite abstraction of our original system, namely a system that preserves some properties of interest of the initial system, while ignoring detail. Results in this direction for the nonlinear centralized case have been obtained in the papers \cite{PgGaTp08}, \cite{ZmPgMmTp12} where the notions of approximate simulation and bisimulation are exploited for certain classes of nonlinear systems under appropriate stability assumptions. The notion of bisimulation, which has its origin in computer science, (see for instance \cite{BcKjp08}, Chapter 7) and refers to transition systems, guarantees that if the initial system and its abstraction are bisimilar, then the task of checking feasibility of high level plans for the original system reduces to the same task for its abstraction and vice versa.

Another tool towards this direction is the hybridization approach \cite{AeDtGa07}, where the behaviour of a nonlinear system is abstracted by means of a piecewise affine hybrid system on simplices. Motion planing techniques for the later case have been developed in the recent works \cite{GaMs08}, \cite{GaMs12}.

In our framework, we focus on multi-agent systems and assume that the agents' dynamics consist of feedback interconnection terms, which ensure that certain system properties as for instance connectivity or (and) invariance are preserved, and free input terms, which provide the ability for motion planning under the coupled constraints. In this report, we aim at quantifying admissible space-time discretizations of our system's behaviour which enable us to capture reachability properties of the original system. In those first results we provide sufficient conditions which establish that the abstraction of our original system is well posed. The later implies that the finite transition system which serves as an abstract model of the multi-agent system has at least one outgoing transition for each discrete state.

\section{Preliminaries and Notation}

We use the notation $|x|$ for the Euclidean norm of a vector $x\in\Rat{n}$. For a subset $S$ of $\Rat{n}$, we denote by ${\rm cl}(S)$, ${\rm int}(S)$ and $\partial S$ its closure, interior and boundary, respectively, where $\partial S:={\rm cl}(S)\setminus{\rm int}(S)$. Given $R>0$ and $y\in\Rat{n}$, we denote by $B(R)$ the closed ball with center $0\in\Rat{n}$ and radius $R$, namely $B(R):=\{x\in\Rat{n}:|x|\le R \}$ and $B_x(R):=\{x\in\Rat{n}:|x-y|\le R \}$. Given two sets $A,B\in\Rat{n}$ their Minkowski sum is defined as
\begin{equation*}
A+B:=\{x+y\in\Rat{n}:x\in A, y\in B\}
\end{equation*}

Consider a multi-agent system with $N$ agents. For each agent $i\in\{1,\ldots,N\}$ we use the notation $\mathcal{N}_i$ for the set of its neighbors and $\Ni$ for its cardinality. We also consider an ordering of the agent's neighbors which we denote by $j_1,\ldots,j_{\Ni}$. Given an index set $\I$ and an agent $i\in\{1,\ldots,N\}$ with neighbors $j_1,\ldots,j_{\Ni}\in\{1,\ldots,N\}$, we define the mapping ${\rm pr}_i:\I^N\to\I^{\Ni+1}$ which assigns to each $N$-tuple $(l_1,\ldots,l_N)\in\I^N$ the $\Ni+1$-tuple $(l_1,l_{j_1},\ldots,l_{j_{\Ni}})\in\I^{\Ni+1}$.

We proceed by providing a formal definition for the notion of a transition system (see for instance \cite{BcKjp08}, \cite{Pg03}, \cite{PgGaTp08}).

\begin{dfn}
A transition system is a quintuple $TS:=(Q,L,\longrightarrow,O,H)$, where:

\textbullet\; $Q$ is a set of states.

\textbullet\; $L$ is a set of actions.

\textbullet\; $\longrightarrow$ is a transition relation with $\longrightarrow\subset Q\times L\times Q$.

\textbullet\; $O$ is an output set.

\textbullet\; $H$ is an output function from $Q$ to $O$.

\noindent The transition system is said to be finite, if $Q$ and $L$ are finite sets. We also use the (standard) notation $q\overset{l}{\longrightarrow} q'$ to denote an element $(q,l,q')\in\longrightarrow$. For every $q\in Q$ and $l\in L$ we use the notation ${\rm Post}(q;l):=\{q'\in Q:(q,l,q')\in\longrightarrow\}$.
\end{dfn}

\noindent We have adopted the definition from \cite{PgGaTp08} with the modification of naming the elements of the set $L$ actions (see \cite{BcKjp08}, Ch. 2) instead of labels as in \cite{PgGaTp08}. We will clarify this choice in the next section.

\section{Abstractions for Multi-Agent Systems}

We focus on multi-agent systems with single integrator dynamics
\begin{equation}\label{single:integrator}
\dot{x}_{i}=u_{i},x_{i}\in\Rat{n},i=1,\ldots,N
\end{equation}

\noindent and consider as inputs decentralized control laws of the form
\begin{equation}\label{general:feedback:law}
u_{i}=f_{i}(x_{i},x_{j_{1}},\ldots,x_{j_{|\mathcal{N}_{i}|}})+v_{i}, i=1,\ldots,N
\end{equation}

\noindent consisting of two terms, the feedback term $f_{i}(\cdot)$ which depends on the states of $i$ and its neighbors, and the free input $v_{i}$. We assume that for each $i=1,\ldots,N$ it holds $x_{i}\in D$ where $D$ is a domain of $\Rat{n}$ and that each $f_{i}(\cdot)$ is locally Lipschitz. We also assume that $v_{i}\in\mathcal{U}_{i}$, $i=1,\ldots,N$ where $\mathcal{U}_{i}$ is a bounded subset of $L^{\infty}(\RgeO;\Rat{n})$ for each $i$ and define $\mathcal{U}:=\mathcal{U}_{1}\times\cdots\times\mathcal{U}_{N}$.

In order to justify our subsequent analysis, we assume that the $f_i$'s are globally bounded and that the maximum magnitude of the feedback terms is higher than that of the free inputs, since we are primarily interested in maintaining the property that the feedback is designed for and, secondarily, in exploiting the free inputs in order to accomplish high level tasks. In what follows, we consider a cell decomposition of the state space $D$ (which can be regarded as a partition of $D$) and a time discretization step $\delta t>0$. In particular, we adopt a modification of the corresponding definition from \cite[p 129-called cell covering]{Gl02}.

\begin{dfn} \label{cell:decomposition}
Let $D$ be a domain of $\Rat{n}$. A cell decomposition $\mathcal{S}=\{S_{l}\}_{l\in\mathcal{I}}$ of $D$, where $\mathcal{I}$ is a finite or coutable index set, is a finite or countable family of uniformly bounded sets $S_{l}$, $l\in\mathcal{I}$ whose interior is a domain, such that ${\rm int}(S_{l})\cap {\rm int}(S_{\hat{l}})=\emptyset$ for all $l\ne\hat{l}$ and $\cup_{l\in\mathcal{I}} S_{l}=D$.
\end{dfn}

\noindent Our ultimate goal is to define finite abstractions for closed loop multi-agent systems of the form  \eqref{single:integrator}-\eqref{general:feedback:law} which evolve inside a bounded domain and satisfy the following \textbf{invariance assumption}.

\noindent \textbf{IA.} For every initial condition $x(0)\in A$ of system  \eqref{single:integrator}-\eqref{general:feedback:law} where $A$ is an appropriate subset of $D$ and every free input $v\in\mathcal{U}$ there exists a unique solution for \eqref{single:integrator}-\eqref{general:feedback:law} which is defined for all $t\ge 0$ and remains in $D$ (for all $t\ge 0$).

A motivating example for this framework has been studied in our companion work \cite{BdDd15} where network connectivity as well as invariance of the system's solution inside a bounded domain and robustness of those properties with respect to free inputs are guaranteed for the single integrator model. A finite cell decomposition in that case can lead to a finite transition system which captures the properties of interest of the multi-agent system and hence enables the investigation for computable solutions with respect to high level plan specifications.

A basic feature that we want to satisfy through our space and time discretization is the possibility to maintain some of the reachability properties of the nonlinear system, when we consider the finite transition system that results from the cell decomposition and the time discretization. Informally, we would like to consider for each agent $i$ the transition system with states the possible modes of the cell decomposition, namely the cells of the state partition, labels all the possible cells of the agents neighbors and transition relation defined in the sense that a final cell is reachable from an initial one, if for all states in the initial cell there is a free input such that the solution of the system will reach the final cell at time $\delta t$ for all possible initial states of the agents neighbors and their corresponding free inputs. Feasibility of high level plans requires the corresponding system to be well posed-meaningful, which implies that for each initial cell it is possible to transit to (at least) one final cell.

One main challenge in the attempt to provide meaningful decentralized abstractions even in this fully actuated with respect to the free inputs case is the interconnection between the agents through the $f_{i}(\cdot)$ terms. The later leads us to reformulate our informal consideration above and motivates us to define appropriate hybrid feedback laws in the place of the $v_{i}$'s which will guarantee our desired well posed transitions. Before proceeding to the necessary definitions related to our problem formulation, we provide some bounds on the dynamics of the multi agent system. In order to simplify the subsequent analysis, which we aim to appropriately modify in order to include domains satisfying (IA) and hence extract finite transition systems, we assume for \eqref{single:integrator}-\eqref{general:feedback:law} that $D=A=\Rat{n}$.

We also assume that the feedback terms $f_{i}(\cdot)$ are globally bounded, namely, there exists a constant $M>0$ such that
\begin{equation} \label{dynamics:bound}
|f_{i}(x_{i},x_{j_{1}},\ldots,x_{j_{|\mathcal{N}_{i}|}})|\le M, \forall (x_{i},x_{j_{1}},\ldots,x_{j_{|\mathcal{N}_{i}|}})\in\Rat{(|\mathcal{N}_{i}|+1)n} \\
\end{equation}

\noindent Furthermore, instead of considering arbitrary input sets $\mathcal{U}_{i}$, $i=1,\ldots,N$ we require that the free inputs $v_{i}$ satisfy the bound
\begin{equation}\label{input:bound}
|v_{i}(t)|\le v_{\max},\forall t\ge 0, i=1,\ldots,N
\end{equation}

\noindent Given the time step $\delta t$, and the bounds $M$ and $v_{\max}$ on the feedback and input terms, we introduce the following lengthscale
\begin{equation} \label{Rmax}
R_{\max}:=\delta t(M+v_{\max})
\end{equation}

\noindent with $M$ and $v_{\max}$ as given in \eqref{dynamics:bound} and \eqref{input:bound}, respectively. It follows from \eqref{single:integrator}-\eqref{general:feedback:law}, \eqref{dynamics:bound}, \eqref{input:bound} and \eqref{Rmax} that $R_{\max}$ is the maximum distance an agent can cross within time $\delta t$.

Given a cell decomposition $\S :=\{S_l\}_{l\in\I}$ of $\Rat{n}$, we frequently use the notation $\tilde{l}_i=(l_i,l_i^1,\ldots,$ $l_i^{\Ni})\in\I^{\Ni+1}$ to denote the indices of the cells where agent $i$ and its neighbors belong at a certain time instant (usually at $t=0$). We also refer to $\tilde{l}_i$ as a (initial) cell configuration of agent $i$. Similarly, we use the notation $\bar{l}=(\bar{l}_{1},\ldots,\bar{l}_{N})\in\mathcal{I}^{N}$ to specify the indices of the cells where all the $N$ agents belong at a given time instant. Thus, given a cell configuration $\bar{l}$ we can determine the cell configuration $\tilde{l}_i$ of agent $i$ through the mapping ${\rm pr}_i:\I^N\to\I^{\Ni+1}$, namely $\tilde{l}_i={\rm pr}_i(\bar{l})$ (see Notations). In this report, we are primarily interested in the evolution of the system on the time interval $[0,\delta t]$, since we focus on the transitions
from initial states at $t=0$ to final states at $t=\delta t$. Thus, we will also use the term final cell configuration when referring to the time instant $\delta t$.

Before defining the notion of a well posed space time discretization we provide a class of hybrid feedback laws, parameterized by the agents initial conditions, which we assign to the free inputs $v_i$ in order to obtain meaningful discrete transitions.

\begin{dfn}
\noindent Given a space-time discretization $\mathcal{S}-\delta t$ ($\S :=\{S_l\}_{l\in\I}$) an agent $i\in\{1,\ldots,N\}$ and an initial cell configuration
\begin{equation*}
\tilde{l}_i=(l_i,l_i^1,\ldots,l_i^{\Ni})\in\I^{\Ni+1}
\end{equation*}

\noindent of $i$, we say that the mapping
\begin{equation*}
[0,T)\times\Rat{(\Ni+1)n}\times\Rat{n}\ni(t,x_{i},x_{j_{1}},\ldots,x_{j_{|\mathcal{N}_{i}|}};x_{i0})\to k_{i,\tilde{l}_i}(t,x_{i},x_{j_{1}},\ldots,x_{j_{|\mathcal{N}_{i}|}};x_{i0})\in\Rat{n}
\end{equation*}

\noindent satisfies property \textbf{(P)}, if the following hold.

\noindent\textbf{(P1)} $T>\delta t$.

\noindent\textbf{(P2)} For each $x_{i0}\in\Rat{n}$ the mapping $k_{i,\tilde{l}_i}(\cdot;x_{i0}):[0,T)\times\Rat{(\Ni+1)n}\to\Rat{n}$ is locally Lipschitz continuous.

\noindent\textbf{(P3)} It holds
\begin{align}
|k_{i,\tilde{l}_i}(t,x_{i},x_{j_{1}},\ldots,x_{j_{|\mathcal{N}_{i}|}};x_{i0})| & \le v_{\max},\forall t\in[0,\delta t], \nonumber \\
x_{i}\in S_{l_i}+B(R_{\max}), x_{j_\kappa}\in S_{l_i^{\kappa}} & +B(R_{\max}),\kappa=1,\ldots,\Ni,x_{i0}\in S_{l_i} \label{feedback:k:bound}
\end{align}

\noindent with $v_{\max}$ as given in \eqref{input:bound} and $R_{\max}$ as in \eqref{Rmax}.
\end{dfn}

\noindent We following provide the definition of a well posed space-time discretization, in accordance to our previous discussions.

\begin{dfn}\label{well:posed:discretization}
Consider a cell decomposition $\S=\{S_l\}_{l\in\I}$ of $\Rat{n}$ and a time step $\delta t$.

\noindent \textbf{(a)} Given an agent $i\in\{1,\ldots,N\}$, an initial cell configuration $\tilde{l}_i=(l_i,l_i^1,\ldots,l_i^{\Ni})\in\I^{\Ni+1}$ of $i$ and a cell index $l_i'\in\I$ we say that the transition $l_i\overset{\tilde{l}_i}{\longrightarrow}l_i'$ is well posed with respect to the space-time discretization $\S-\delta t$ if there exists a feedback law
\begin{equation} \label{feedback:for:i}
k_{i,\tilde{l}_i}(t,x_{i},x_{j_{1}},\ldots,x_{j_{|\mathcal{N}_{i}|}};x_{i0})
\end{equation}

\noindent parameterized by $x_{i0}\in\Rat{n}$ (the initial condition of $i$) and satisfying property (P), such that condition (C) below is fulfilled.

\noindent \textbf{(C)} For each initial cell configuration $\bar{l}=(\bar{l}_{1},\ldots,\bar{l}_{N})\in\mathcal{I}^{N}$ with ${\rm pr}_i(\bar{l})=\tilde{l}_i$ and for all $\hat{i}\in\{1,\ldots,N\}\setminus \{{i}\}$ and feedback laws
\begin{equation} \label{feedback:for:others}
k_{\hat{i},\tilde{l}_{\hat{i}}}(t,x_{\hat{i}},x_{\hat{j}_{1}},\ldots,x_{\hat{j}_{|\mathcal{N}_{i}|}};x_{\hat{i}0})
\end{equation}

\noindent parmeterized by $x_{\hat{i}0}\in\Rat{n}$ (the initial condition of $\hat{i}$) and satisfying property (P), (with $\tilde{l}_{\hat{i}}={\rm pr}_{\hat{i}}(\bar{l})$) the solution of the closed loop system \eqref{single:integrator}-\eqref{general:feedback:law}-\eqref{feedback:for:i}-\eqref{feedback:for:others} (with $v_{\kappa}=k_{\kappa,\tilde{l}_{\kappa}}$, $\kappa=1,\ldots,N$) satisfies

\begin{equation*}
x_{i}(\delta t,x(0))\in S_{l_i'}
\end{equation*}
\noindent for all initial conditions $x(0)\in\Rat{n}$ with $x_{\kappa}(0)=x_{\kappa 0}\in S_{\bar{l}_{\kappa}}$, $\kappa=1,\ldots,N$.

\noindent \textbf{(b)} We say that the space-time discretization $\S-\delta t$ is well posed if for each agent $i\in\{1,\ldots,N\}$ and cell configuration $\tilde{l}_i=(l_i,l_i^1,\ldots,l_i^{\Ni})\in\I^{\Ni+1}$ of $i$, there exists a cell index $l_i'\in\I$ such that the transition $l_i\overset{\tilde{l}_i}{\longrightarrow}l_i'$ is well posed with respect to $\S-\delta t$.
\end{dfn}

Given a space-time discretization $\S-\delta t$ and based on Definition \ref{well:posed:discretization}(a), we are in a position to provide an exact description of the discrete transition system which serves as an abstract model for the behaviour of each agent. At this point, we do not focus on the output set and map of the transition system and just provide the definition of its state set, label set and transition relation. In particular, for each agent $i$, we define the discrete transition system $TS_i:=(Q,L_i,\longrightarrow_i)$ with state set $Q$ the indices $\I$ of the cell decomposition, actions all possible cell indices of $i$ and its neighbors, namely $L_i:=I^{\Ni+1}$ (the set of all possible cell configurations of $i$) and transition relation $\longrightarrow_i\subset Q\times L_i\times Q$ defined as follows. For any $\hat{l}_i,\hat{l}_i'\in Q$ and $\tilde{l}_i=(l_i,l_i^1,\ldots,l_i^{\Ni})\in\I^{\Ni+1}$
$$
\hat{l}_i\overset{\tilde{l}_i}{\longrightarrow_i}\hat{l}_i'
$$

\noindent iff
$$
\hat{l}_i=l_i\quad{\rm and}\quad l_i\overset{\tilde{l}_i}{\longrightarrow}\hat{l}_i'\quad\textup{is well posed}
$$

\noindent We have preferred to use the term actions instead of labels for the elements of the set $L_i$, because the cell configuration of $i$ indicates how the feedback term $f_i(\cdot)$ acts on-affects the possible transitions of agent $i$.

According to Definition \ref{well:posed:discretization}, a well posed space-time discretization requires the existence of a well posed transition for each agent $i$ and the latter reduces to the selection of an appropriate feedback controller for $i$, which also satisfies Property (P), and the requirement that the other agents also satisfy (P). Yet, it is not completely evident, that given an initial cell configuration and a well posed transition for each agent, that we can choose a feedback law for each, so that the resulting closed loop system will guarantee all these transitions (for all possible initial conditions in the cell configuration). The following proposition clarifies this point.

\begin{prop}\label{discrete:transitions:result}
Consider system \eqref{single:integrator}-\eqref{general:feedback:law}, let $\bar{l}=(\bar{l}_{1},\ldots,\bar{l}_{N})\in\mathcal{I}^{N}$ be an initial cell configuration and assume that the space-time discretization $S-\delta t$ is well posed, which implies that for all $i=1,\ldots,N$ it holds ${\rm Post}_i(\bar{l}_{i};{\rm pr}_{i}(\bar{l}))\ne\emptyset$ (${\rm Post}_i(\cdot)$ refers to the transition system $TS_i$ of each agent -see also Notations). Then for every final cell configuration
\begin{equation}\label{final:cc}
\bar{l}'=(\bar{l}_1',\ldots,\bar{l}_N')\in{\rm Post}_i(\bar{l}_1;{\rm pr}_1(\bar{l}))\times\cdots\times{\rm Post}_i(\bar{l}_N;{\rm pr}_N(\bar{l}))
\end{equation}

\noindent there exist feedback laws
\begin{equation} \label{feedback:for:all}
k_{i,{\rm pr}_i(\bar{l})}(t,x_{i},x_{j_{1}},\ldots,x_{j_{|\mathcal{N}_{i}|}};x_{i0}),i=1,\ldots,N
\end{equation}

\noindent satisfying property (P) and such that for each i=1,\ldots,N the $i$-th component of the solution of the closed loop system \eqref{single:integrator}-\eqref{general:feedback:law}-\eqref{feedback:for:all}  (with $v_{\kappa}=k_{\kappa,{\rm pr}_{\kappa}(\bar{l})}$, $\kappa=1,\ldots,N$) satisfies
\begin{equation}\label{contoler:compatibility}
x_{i}(\delta t,x(0))\in S_{\bar{l}_i'}, \forall x(0)\in\Rat{Nn}:x_{\kappa}(0)=x_{\kappa 0}\in S_{\bar{l}_{\kappa}},\kappa=1,\ldots,N
\end{equation}
\end{prop}

\begin{proof}
Indeed, consider a final cell configuration $\bar{l}'=(\bar{l}_1',\ldots,\bar{l}_N')$ as in \eqref{final:cc} and select for each agent $i\in\{1,\ldots,N\}$ a control law $k_{i,{\rm pr}_i(\bar{l})}(\cdot)$ which ensures that $\bar{l}_i\overset{{\rm pr}_i(\bar{l})}{\longrightarrow_i}\bar{l}_i'$ is well posed. It follows from Definition \ref{well:posed:discretization}(a) that all the feedback laws $k_{i,{\rm pr}_i(\bar{l})}(\cdot)$, $i=1,\ldots,N$ satisfy Property (P) and hence, from Condition (C), that for each $i=1,\ldots,N$ the $i$-th component of the solution of the closed loop system satisfies \eqref{contoler:compatibility}.
\end{proof}

The result of the following proposition guarantees that the selection of the controllers introduced in Definition \ref{well:posed:discretization} provide well posed solutions for the closed loop system on the time interval $[0,\delta t]$. We exploit this result in Proposition \ref{admissible:discretizations} where we derive sufficient conditions for well posed space-time discretizations. Furthermore, Proposition \ref{completeness:result} guarantees that the magnitude of the hybrid feedback laws does not exceed the maximum allowed magnitude of the free inputs $v_{\max}$ on $[0,\delta t]$ and hence establishes consistency with our initial design requirement.

\begin{prop} \label{completeness:result}
Consider the space-time discretization $\S-\delta t$ corresponding to the cell decomposition $\S$ of $\Rat{n}$ and the time step $\delta t$. Let $\bar{l}=(\bar{l}_{1},\ldots,\bar{l}_{N})\in\mathcal{I}^{N}$ be an initial cell configuration and consider the feedback laws
\begin{equation} \label{feedback:for:all2}
k_{i,{\rm pr}_i(\bar{l})}(t,x_{i},x_{j_{1}},\ldots,x_{j_{|\mathcal{N}_{i}|}};x_{i0}),i=1,\ldots,N
\end{equation}

\noindent assigned to the agents that satisfy property (P). Then for all initial conditions $x(0)\in\Rat{n}$ with $x_{i}(0)=x_{i0}\in S_{\bar{l}_i}$, $i=1,\ldots,N$ the solution of the closed loop system \eqref{single:integrator}-\eqref{general:feedback:law}-\eqref{feedback:for:all2} (with $v_i=k_{i,{\rm pr}_i(\bar{l})}$, $i=1,\ldots,N$) is defined on $[0,\delta t]$ and each component $x_i(\cdot)$, $i=1,\ldots,N$ of the solution satisfies
\begin{equation} \label{xis:in:initialcell:plus:Rmax}
x_i(t)\in S_{\bar{l}_i}+B(R_{\max}),\forall t\in[0,\delta t)
\end{equation}

\noindent Hence, it follows from \eqref{xis:in:initialcell:plus:Rmax}, (P3) and continuity of the solution $x(\cdot)$ that

\begin{equation}
|k_{i,{\rm pr}_i(\bar{l})}(t,x_{i}(t),x_{j_{1}}(t),\ldots,x_{j_{|\mathcal{N}_{i}|}}(t);x_{i0})|\le v_{\max},\forall t\in[0,\delta t], i=1,\ldots,N
\end{equation}

\noindent which provides the desired consistency with our design requirement \eqref{input:bound} on the $v_i$'s.
\end{prop}

\begin{proof}
Let $x(0)\in\Rat{Nn}$ with $x_{i}(0)\in S_{\bar{l}_i}$, $i=1,\ldots,N$ be the initial condition of the closed loop system. Then it follows from the local Lipschitz property on the functions $f_{i}(\cdot)$ and the corresponding property on the mappings $k_{i,{\rm pr}_i(\bar{l})}(\cdot;x_{i0})$ provided by (P2), that there exists a unique solution $x(\cdot)=x(\cdot,x(0))$ to the initial value problem defined on the right maximal interval of existence $[0,T_{\max})$. We proceed by proving that \eqref{xis:in:initialcell:plus:Rmax} holds as well (this also implies that $T_{\max}>\delta t$). Indeed, suppose on the contrary that \eqref{xis:in:initialcell:plus:Rmax} is violated and hence, that there exists $\hat{i}\in\{1,\ldots,N\}$ and a time $T$ with
\begin{equation} \label{xhati:atT}
T\in (0,\delta t)\;{\rm and}\;x_{\hat{i}}(T)\notin S_{\bar{l}_{\hat{i}}}+B(R_{\max})
\end{equation}

\noindent By exploiting continuity of $x(\cdot)$ we may define
\begin{equation} \label{time:tau}
\tau:=\max\{\bar{t}\in [0,T]:x_{i}(t)\in{\rm cl}(S_{\bar{l}_{i}}+B(R_{\max})),\forall t\in [0,\bar{t}],i=1,\ldots,N\}
\end{equation}

\noindent Then, it follows from \eqref{xhati:atT} and \eqref{time:tau} that there exists $\tilde{i}\in\{1,\ldots,N\}$ such that
\begin{equation} \label{xtildei:at:tau}
x_{\tilde{i}}(\tau)\in\partial(S_{\bar{l}_{\tilde{i}}}+B(R_{\max}))
\end{equation}

\noindent and that
\begin{equation} \label{tau:vs:deltat}
\tau<T\le\delta t
\end{equation}

\noindent It also follow from \eqref{time:tau} that
\begin{equation}
x_{i}(t)\in{\rm cl}(S_{\bar{l}_{i}}+B(R_{\max})),\forall t\in [0,\tau],i=1,\ldots,N
\end{equation}

\noindent and thus from property (P3) and continuity of $x(\cdot)$ and $k_{\tilde{i},{\rm pr}_{\tilde{i}}(\bar{l})}(\cdot;x_{\tilde{i}0})$ that
\begin{equation} \label{feedback:for:tildei}
|k_{\tilde{i},{\rm pr}_{\tilde{i}}(\bar{l})}(t,x_{\tilde{i}}(t),x_{\tilde{j}_{1}}(t),\ldots,x_{\tilde{j}_{|\mathcal{N}_{\tilde{i}}|}}(t);x_{\tilde{i}0})|\le v_{\max},\forall t\in[0,\tau]
\end{equation}

\noindent Hence, we get from \eqref{single:integrator}-\eqref{general:feedback:law}, \eqref{Rmax}, \eqref{feedback:for:all2}, \eqref{feedback:for:tildei} and \eqref{tau:vs:deltat} that
\begin{align*}
|x_{\tilde{i}}(\tau)-x_{i0}|&=\left|\int_{0}^{\tau}f_{\tilde{i}}(x_{\tilde{i}}(s),x_{\tilde{j}_{1}}(s),\ldots,x_{\tilde{j}_{|\mathcal{N}_{\tilde{i}}|}}(s))+k_{\tilde{i},{\rm pr}_{\tilde{i}}(\bar{l})}(s,x_{\tilde{i}}(s),x_{\tilde{j}_{1}}(s),\ldots,x_{\tilde{j}_{|\mathcal{N}_{\tilde{i}}|}}(s);x_{\tilde{i}0})ds\right| \\
&\le \int_{0}^{\tau}|f_{\tilde{i}}(x_{\tilde{i}}(s),x_{\tilde{j}_{1}}(s),\ldots,x_{\tilde{j}_{|\mathcal{N}_{\tilde{i}}|}}(s))|+|k_{\tilde{i},{\rm pr}_{\tilde{i}}(\bar{l})}(s,x_{\tilde{i}}(s),x_{\tilde{j}_{1}}(s),\ldots,x_{\tilde{j}_{|\mathcal{N}_{\tilde{i}}|}}(s);x_{\tilde{i}0})|ds \\
& \le \int_{0}^{\tau}(M+v_{\max})ds=\tau(M+v_{\max})<\delta t(M+v_{\max})=R_{\max}
\end{align*}

\noindent It thus follows from Fact I in the Appendix that $x_{\tilde{i}}(\tau)\notin \partial(S_{\bar{l}_{\tilde{i}}}+B(R_{\max}))$ which contradicts \eqref{xtildei:at:tau} and the proof is complete.
\end{proof}

\section{Admissible Space-Time Discretizations}

We proceed by providing some extra details for the dynamics as determined by the feedback law in \eqref{general:feedback:law}. Assuming that the $f_i$'s are globally Lipschitz functions it follows that there exists a constant $L>0$ such that
\begin{align*}
|f_{i}(x_{i},x_{j_{1}},\ldots,x_{j_{|\mathcal{N}_{i}|}})-f_{i}(y_{i},y_{j_{1}},\ldots,y_{j_{|\mathcal{N}_{i}|}})|&\le L|(x_{i},x_{j_{1}},\ldots,x_{j_{|\mathcal{N}_{i}|}})-(y_{i},y_{j_{1}},\ldots,y_{j_{|\mathcal{N}_{i}|}})|, \\
\forall (x_{i},x_{j_{1}},\ldots,x_{j_{|\mathcal{N}_{i}|}}),(y_{i},y_{j_{1}},\ldots,y_{j_{|\mathcal{N}_{i}|}}) & \in\Rat{(|\mathcal{N}_{i}|+1)n}
\end{align*}

\noindent Furthermore, if we want to achieve more accurate bounds for the dynamics of the feedback controllers we assign to the free inputs $v_i$ (those will be clarified in the proof of Proposition \ref{admissible:discretizations}), we can choose (posssibly) different Lipschitz constants $L_{1},L_{2}>0$ such that
\begin{align}
|f_{i}(x_{i},x_{j_{1}},\ldots,x_{j_{|\mathcal{N}_{i}|}})-f_{i}(x_{i},y_{j_{1}},\ldots,y_{j_{|\mathcal{N}_{i}|}})| & \le L_{1}|(x_{i},x_{j_{1}},\ldots,x_{j_{|\mathcal{N}_{i}|}})-(x_{i},y_{j_{1}},\ldots,y_{j_{|\mathcal{N}_{i}|}})|,  \nonumber \\
\forall x_{i}\in\Rat{n},(x_{j_{1}},\ldots,x_{j_{|\mathcal{N}_{i}|}}),(y_{j_{1}},\ldots,y_{j_{|\mathcal{N}_{i}|}}) & \in \Rat{|\mathcal{N}_{i}|n} \label{dynamics:bound1} \\
|f_{i}(x_{i},x_{j_{1}},\ldots,x_{j_{|\mathcal{N}_{i}|}})-f_{i}(y_{i},x_{j_{1}},\ldots,x_{j_{|\mathcal{N}_{i}|}})| & \le L_{2}|(x_{i},x_{j_{1}},\ldots,x_{j_{|\mathcal{N}_{i}|}})-(y_{i},x_{j_{1}},\ldots,x_{j_{|\mathcal{N}_{i}|}})|, \nonumber \\
\forall x_{i},y_{i} \in\Rat{n},(x_{j_{1}},\ldots,x_{j_{|\mathcal{N}_{i}|}}) & \in \Rat{|\mathcal{N}_{i}|n} \label{dynamics:bound2}
\end{align}

\noindent In order to provide some extra motivation on considering both constants $L_1$ and $L_2$, we note that in order to derive sufficient conditions for a well posed discretization, we design for each agent $i$ inside a cell $S_{l_i}$ a feedback, in order to ``track" a given reference trajectory (of $i$) starting in the same cell. In particular, the constant $L_{1}$ provides bounds on our choice of feedback in order to compensate for the deviation of agent's $i$ dynamics from its corresponding dynamics along the reference trajectory, due to the time evolution of its neighbors states. On the other hand, the constant $L_{2}$ provides bounds on our choice of feedback in order to compensate for the deviation of the initial state with respect to the initial state of the reference trajectory.

In order to apply the previous results it is convenient that we define the least upper bound on the diameter of the cells in $\mathcal{S}$, namely

\begin{equation*}
d_{\max}:=\sup\{\sup\{|x-y|:x,y\in S_{l}\}:l\in\mathcal{I}\}
\end{equation*}

\noindent which due to Definition \ref{cell:decomposition} is well defined. We call $d_{\max}$ the \textbf{diameter} of the cell decomposition.

Consider again system  \eqref{single:integrator}-\eqref{general:feedback:law}, namely the system
\begin{equation} \label{full:information:dynamics}
\dot{x}_{i}=f_{i}(x_{i},x_{j_{1}},\ldots,x_{j_{|\mathcal{N}_{i}|}})+v_{i},i=1,\ldots,N
\end{equation}

We want to determine sufficient conditions relating the Lipschitz constants $L_{1}$, $L_{2}$, and the bounds $M$, $v_{\max}$ of the system's dynamics, as well as the space and time scales $d_{\max}$ and $\delta t$ of the space-time discretization $\mathcal{S}-\delta t$ which guarantee that $\mathcal{S}-\delta t$ is well posed. As discussed at the beginning of the previous section, we require that the bound on the $f_i(\cdot)$ terms is greater than the maximum magnitude of the free inputs and thus impose the additional restriction
\begin{equation} \label{vmax:vs:M}
v_{\max}<M
\end{equation}

\noindent The desired sufficient conditions for a well posed discretization are provided in the following result.

\begin{prop}\label{admissible:discretizations}
Consider a cell decomposition $\S$ of $\Rat{n}$ and a time step $\delta t$. For the multi-agent system \eqref{full:information:dynamics} a sufficient condition which guarantees that the space-time discretization $\mathcal{S}-\delta t$ is well posed, is that the diameter $d_{\max}$ of $\mathcal{S}$ and the time step $\delta t$ satisfy the restrictions

\begin{align}
d_{\max}&\in\left(0,\frac{v_{\max}^{2}}{4M\tilde{L}}\right] \label{dmax:interval}\\
\delta t &\in\left[\frac{v_{\max}-\sqrt{v_{\max}^{2}-4M\tilde{L}d_{\max}}}{2M\tilde{L}},\frac{v_{\max}+\sqrt{v_{\max}^{2}-4M\tilde{L}d_{\max}}}{2M\tilde{L}}\right] \label{deltat:interval}
\end{align}

\noindent with

\begin{align} 
\tilde{L}_{i}:=&2L_{2}+4L_{1}\sqrt{|\mathcal{N}_{i}|} \label{constant:tildeLi}\\
\tilde{L}:=&\max\{\tilde{L}_{i},i=1,\ldots,N\} \label{constant:tildeL}
\end{align}
 
\noindent and where $L_1$ and $L_2$ are given in \eqref{dynamics:bound1} and \eqref{dynamics:bound2}.

In particular, for each agent $i\in\{1,\ldots,N\}$ and cell configuration $\tilde{l}_i=(l_i,l_i^1,\ldots,l_i^{\Ni})\in\I^{\Ni+1}$ of $i$ we select a reference point
\begin{equation}
(x_{i,G},x_{j_{1},G},\ldots,x_{j_{\Ni},G})\in S_{l_{i}}\times S_{l_{i}^1}\times\cdots\times S_{l_i^{|\mathcal{N}_{i}|}}
\end{equation}

\noindent and define the feedback law  $k_{i,\tilde{l}_i}:\RgeO\times\Rat{(\Ni+1)n}\times\Rat{n}\to\Rat{n}$ as
\begin{equation} \label{feedback:ki}
k_{i,\tilde{l}_i}(t,x_i,x_{j_1},\ldots,x_{j_{\Ni}};x_{i0}):=k_{i,\tilde{l}_i,1}(x_i,x_{j_1},\ldots,x_{j_{\Ni}})+k_{i,\tilde{l}_i,2}(x_{i0})+k_{i,\tilde{l}_i,3}(t;x_{i0})
\end{equation}

\noindent where
\begin{align}
k_{i,\tilde{l}_i,1}(x_{i},x_{j_{1}},\ldots,x_{j_{\Ni}}):=-&[f_i(x_i,x_{j_1},\ldots,x_{j_{\Ni}})-f_i(x_i,x_{j_{1},G},\ldots,x_{j_{\Ni},G})] \nonumber \\
&\forall (x_i,x_{j_1},\ldots,x_{j_{\Ni}})\in\Rat{(\Ni+1)n} \label{feedback:ki1} \\
k_{i,\tilde{l}_i,2}(x_{i0}):=-&\frac{1}{\delta t}[x_{i0}-x_{i,G}],\forall x_{i0}\in\Rat{n} \label{feedback:ki2} \\
k_{i,\tilde{l}_i,3}(t;x_{i0}):=-&\left[\tilde{f}_{i,\tilde{l}_i}\left(\tilde{x}_i(t)+\left(1-\frac{t}{\delta t}\right)(x_{i0}-x_{i,G})\right)-\tilde{f}_{i,\tilde{l}_i}(\tilde{x}_{i}(t))\right] \nonumber \\
&\forall t\in\RgeO,x_{i0}\in\Rat{n} \label{feedback:ki3}
\end{align}

\noindent the function $\tilde{f}_{i,\tilde{l}_i}(\cdot)$ is given as
\begin{equation}\label{averaged:dynamics}
\tilde{f}_{i,\tilde{l}_i}(x_i):=f_i(x_i,x_{j_{1},G},\ldots,x_{j_{\Ni},G}),\forall x_i\in\Rat{n}
\end{equation}

\noindent and $\tilde{x}_i(\cdot)$ is the solution of the initial value problem
\begin{equation}\label{reference:solution}
\dot{\tilde{x}}_i=\tilde{f}_{i,\tilde{l}_i}(\tilde{x}_i),\tilde{x}_i(0)=x_{i,G}
\end{equation}

\noindent Then it follows that $k_{i,\tilde{l}_i}(\cdot)$ satisfies Property (P) and that there exists $l_i'\in\I$ such that condition (C) of Definition \ref{well:posed:discretization}(a) is fulfilled. In particular we choose $l_i'$ such that $\tilde{x}_i(\delta t)\in S_{l_i'}$.
\end{prop}

\begin{proof}
In order to prove our result, we want to show that the requirements of Definition \eqref{well:posed:discretization}(b) are fulfilled. Let $\mathcal{S}=\{S_{l}\}_{l\in\mathcal{I}}$ be a cell decomposition of $\Rat{n}$ with maximum diameter $d_{\max}$ and consider a time step $\delta t$, such that \eqref{dmax:interval} and \eqref{deltat:interval} hold. We want to show that for each $i=1,\ldots,N$ and $\tilde{l}_i=(l_i,l_i^1,\ldots,l_i^{\Ni})\in\I^{\Ni+1}$ there exists a cell index $l_i'\in\I$ such that the transition $l_i\overset{\tilde{l}_i}{\longrightarrow}l_i'$ is well posed with respect to $\S-\delta t$. Pick $i\in\{1,\ldots,N\}$ and $\tilde{l}_i=(l_i,l_i^1,\ldots,l_i^{\Ni})\in\I^{\Ni+1}$.  In order to find $l_i'\in\I$ such that $l_i\overset{\tilde{l}_i}{\longrightarrow}l_i'$ is well posed, we need according to Definition \ref{well:posed:discretization}(a) to find a feedback law \eqref{feedback:for:i} satisfying Property (P) and in such a way that condition (C) is fulfilled. We brake the proof in three steps.

\noindent \textbf{STEP 1: Selection of the feedback $k_{i,\tilde{l}_i}(\cdot)$ and estimation of bounds on $k_{i,\tilde{l}_i,1}(\cdot)$, $k_{i,\tilde{l}_i,2}(\cdot)$ and $k_{i,\tilde{l}_i,3}(\cdot)$ as given in \eqref{feedback:ki1}-\eqref{feedback:ki3}.}

In this step, we use the notation $x$ for a vector $(x_{i},x_{j_{1}},\ldots,x_{j_{|\mathcal{N}_{i}|}})\in\Rat{(|\mathcal{N}_{i}|+1)n}$ and $\bar{x}$ for its projection to its last $|\mathcal{N}_{i}|n$ coordinates, namely, $\bar{x}:=(x_{j_{1}},\ldots,x_{j_{|\mathcal{N}_{i}|}})\in\Rat{|\mathcal{N}_{i}|n}$. As in the statement of the proposition we select an arbitrary reference point $x_{G}(=x_{G,\tilde{l}})=(x_{i,G},\bar{x}_{G})=(x_{i,G},x_{j_{1},G},\ldots,x_{j_{|\mathcal{N}_{i}|},G})\in S_{l_{i}}\times S_{l_i^1}\times\cdots\times S_{l_i^{|\mathcal{N}_{i}|}}$. Then for all $x\in\Rat{(|\mathcal{N}_{i}|+1)n}$ we have
\begin{align} \label{fi:difference}
f_{i}(x) & =f_{i}(x_{i},\bar{x})=f_{i}(x_{i},\bar{x}_{G})+f_{i}(x_{i},\bar{x})-f_{i}(x_{i},\bar{x}_{G})\iff \nonumber \\
f_{i}(x) & =f_{i}(x_{i},\bar{x}_{G})+\Delta_{i,\tilde{l}}(x_{i},\bar{x})
\end{align}

\noindent where
\begin{equation} \label{Deltaix:dfn}
\Delta_{i,\tilde{l}_i}(x_{i},\bar{x}):= f_{i}(x_{i},\bar{x})-f_{i}(x_{i},\bar{x}_{G})
\end{equation}

\noindent We following show that
\begin{equation} \label{dix:bound}
|\Delta_{i,\tilde{l}_i}(x_{i},\bar{x})|\le L_{1}\sqrt{|\mathcal{N}_{i}|}(R_{\max}+d_{\max})
\end{equation}

\noindent for all $x_i\in\Rat{n}$ and $\bar{x}=(x_{j_{1}},\ldots,x_{j_{|\mathcal{N}_{i}|}})$ satisfying
\begin{equation} \label{distance:neighbor:to:simplex}
x_{j_{\kappa}}\in S_{l_{\kappa}}+B(R_{\max}), \kappa=1,\ldots,|\mathcal{N}_{i}|
\end{equation}

\noindent Indeed, let $\bar{x}$ satisfying \eqref{distance:neighbor:to:simplex}. Then for each $\kappa=1,\ldots,|\mathcal{N}_{i}|$ there exists $\tilde{x}_{j_{\kappa}}$ with
\begin{equation} \label{xjtilde:properties}
\tilde{x}_{j_{\kappa}}\in S_{l_{\kappa}}\quad{\rm and}\quad |\tilde{x}_{j_{\kappa}}-x_{j_{\kappa}}|\le R_{\max}
\end{equation}

\noindent Hence, from \eqref{Deltaix:dfn}, \eqref{xjtilde:properties} and  \eqref{dynamics:bound1} we get
\begin{align*}
|\Delta_{i,\tilde{l}_i}(x_{i},\bar{x})|&=|f_{i}(x_{i},\bar{x})-f_{i}(x_{i},\bar{x}_{G})|\le L_{1}|\bar{x}-\bar{x}_{G}| \\
&=L_{1}|(x_{j_{1}}-x_{j_{1},G},\ldots,x_{j_{|\mathcal{N}_{i}|}}-x_{j_{|\mathcal{N}_{i}|},G})|=L_{1}\left(\sum_{\kappa=1}^{|\mathcal{N}_{i}|}|x_{j_{k}}-x_{j_{k},G}|^{2}\right)^{\frac{1}{2}} \\
&\le L_{1}\left(\sum_{\kappa=1}^{|\mathcal{N}_{i}|}(|x_{j_{\kappa}}-\tilde{x}_{j_{\kappa}}|+|\tilde{x}_{j_{\kappa}}-x_{j_{\kappa},G}|)^{2}\right)^{\frac{1}{2}} \\
&\le L_{1}\left(\sum_{\kappa=1}^{|\mathcal{N}_{i}|}(R_{\max}+d_{\max})^{2}\right)^{\frac{1}{2}}= L_{1}\sqrt{|\mathcal{N}_{i}|}(R_{\max}+d_{\max})
\end{align*}

\noindent In the sequel, we define $f_{i,\tilde{l}_i}(\cdot)$ as in \eqref{averaged:dynamics}. By taking into account  \eqref{averaged:dynamics} and \eqref{fi:difference} it follows that
\begin{equation}\label{fi:Deltai}
f_{i}(x)=f_{i,\tilde{l}_i}(x_{i})+\Delta_{i,\tilde{l}}(x_{i},\bar{x}),\forall x\in\Rat{\Ni+1}
\end{equation}

\noindent and that due to \eqref{dynamics:bound2}, that $\tilde{f}_{i,\tilde{l}_i}(\cdot)$ satisfies the Lipschitz condition
\begin{align}
|\tilde{f}_{i,\tilde{l}_i}(x)-\tilde{f}_{i,\tilde{l}_i}(y)|&=|f_{i}(x,x_{j_{1},G},\ldots,x_{j_{|\mathcal{N}_{i}|},G})-f_{i}(y,x_{j_{1},G},\ldots,x_{j_{|\mathcal{N}_{i}|},G})| \nonumber \\
& \le L_{2}|(x,x_{j_{1},G},\ldots,x_{j_{|\mathcal{N}_{i}|},G})-(y,x_{j_{1},G},\ldots,x_{j_{|\mathcal{N}_{i}|},G})| \nonumber \\
& =L_{2}|(x-y,0,\ldots,0)|=L_{2}|x-y|\Rightarrow  \nonumber \\
|\tilde{f}_{i,\tilde{l}_i}(x)-\tilde{f}_{i,\tilde{l}_i}(y)| & \le L_{2}|x-y| \label{tildefi:Lipschitz:const}
\end{align}

\noindent Now define $k_{i,\tilde{l}_i,1}(\cdot)$, $k_{i,\tilde{l}_i,2}(\cdot)$ and $k_{i,\tilde{l}_i,3}(\cdot)$ as in \eqref{feedback:ki1},  \eqref{feedback:ki2} an \eqref{feedback:ki3}, respectively. By virtue of \eqref{tildefi:Lipschitz:const}, the solution $\tilde{x}_i(\cdot)$ of the initial value problem \eqref{reference:solution} is defined for all $t\ge 0$ and thus $k_{i,\tilde{l}_i,3}(\cdot)$ is well defined. Also, from  \eqref{feedback:ki1} and \eqref{Deltaix:dfn} we have
\begin{equation}\label{ki:Deltai}
k_{i,\tilde{l}_i,1}(x_i,\bar{x})=-\Delta_{i,\tilde{l}_i}(x_{i},\bar{x}),\forall (x_i,\bar{x})\in\Rat{(\Ni+1)n}
\end{equation}

\noindent  Hence, we get from \eqref{dix:bound} and \eqref{distance:neighbor:to:simplex} that
\begin{align}
|k_{i,\tilde{l}_i,1}(x_{i},x_{j_{1}},\ldots,x_{j_{\Ni}})| & \le  L_{1}\sqrt{|\mathcal{N}_{i}|}(R_{\max}+d_{\max}), \nonumber \\
\forall x_i & \in\Rat{n},x_{j_{\kappa}}\in S_{l_{\kappa}}+B(R_{\max}), \kappa=1,\ldots,|\mathcal{N}_{i}| \label{ki1:bound}
\end{align}

\noindent Furthermore, by recalling that $x_{i,G}\in S_{l_i}$, it follows directly from \eqref{feedback:ki2} that
\begin{equation} \label{ki2:bound}
|k_{i,\tilde{l}_i,2}(x_{i0})|=\frac{1}{\delta t}|x_{i0}-x_{i,G}|\le \frac{1}{\delta t}d_{\max},\forall x_{i0}\in S_{l_i}
\end{equation}

\noindent and from \eqref{tildefi:Lipschitz:const} and \eqref{feedback:ki3} that
\begin{align}
|k_{i,\tilde{l}_i,3}(t;x_{i0})|&=\left|\tilde{f}_{i,\tilde{l}_i}\left(\tilde{x}_i(t)+\left(1-\frac{t}{\delta t}\right)(x_{i0}-x_{i,G})\right)-\tilde{f}_{i,\tilde{l}_i}(\tilde{x}_i(t))\right| \nonumber \\
&\le L_2\left|\left(\tilde{x}_i(t)+\left(1-\frac{t}{\delta t}\right)(x_{i0}-x_{i,G})\right)-\tilde{x}_i(t)\right| \nonumber \\
&\le L_{2}|x_{i0}-x_{i,G}|\le L_{2}d_{\max},\forall t\in[0,\delta t],x_{i0}\in S_{l_i} \label{ki3:bound}
\end{align}

\noindent \textbf{STEP 2: Verification of Property (P) for the feedback law \eqref{feedback:ki} for $d_{\max}-\delta t$ satisfying  \eqref{dmax:interval} and \eqref{deltat:interval}.}

In this step we prove that the proposed feedback law \eqref{feedback:ki} satisfies Properties (P1), (P2) and (P3). Verification of (P1) and (P2) is rather straightforward, so we focus on (P3), namely, we show that \eqref{feedback:k:bound} holds. By taking into account \eqref{feedback:ki}, \eqref{ki1:bound}, \eqref{ki2:bound} and \eqref{ki3:bound} it suffices to prove that

\begin{equation} \label{condition:dmax:vmax}
L_{1}\sqrt{|\mathcal{N}_{i}|}(R_{\max}+d_{\max})+\frac{1}{\delta t}d_{\max}+L_{2}d_{\max}\le v_{\max}
\end{equation}

\noindent By recalling \eqref{Rmax} and imposing the additional requirement that
\begin{equation} \label{Rmax:vs:dmax}
\delta t(M+v_{\max})\ge d_{\max}\Rightarrow R_{\max}\ge d_{\max}
\end{equation}

\noindent it suffices instead of \eqref{condition:dmax:vmax} to show that

\begin{equation*}
(2L_{1}\sqrt{|\mathcal{N}_{i}|}+L_{2})R_{\max}+\frac{1}{\delta t}d_{\max}\le v_{\max}
\end{equation*}

\noindent which by virtue of \eqref{Rmax} is equivalent to
\begin{equation}
(M+v_{\max})(2L_{1}\sqrt{|\mathcal{N}_{i}|}+L_{2})\delta t^{2}-v_{\max}\delta t+d_{\max}\le 0 \label{quadratic:condition}
\end{equation}

\noindent By taking into account \eqref{vmax:vs:M}, it suffices instead of \eqref{quadratic:condition} to show that
$$
M(2L_{2}+4L_{1}\sqrt{|\mathcal{N}_{i}|})\delta t^{2}-v_{\max}\delta t+d_{\max}\le 0
$$

\noindent which by virtue of \eqref{constant:tildeLi} is equivalent to
\begin{equation}\label{quadratic:condition2}
M\tilde{L}_{i}\delta t^{2}-v_{\max}\delta t+d_{\max}\le 0 
\end{equation}

\noindent Furthermore, by exploiting \eqref{constant:tildeL} we deduce that \eqref{quadratic:condition2} follows from
\begin{equation} \label{quadratic:condition3}
M\tilde{L}\delta t^{2}-v_{\max}\delta t+d_{\max}\le 0
\end{equation}

\noindent In order for the second order equation \eqref{quadratic:condition3} to have at least one real root (if it has real roots they are positive) it is required that
\begin{equation} \label{dmax:bound}
v_{\max}^{2}-4M\tilde{L}d_{\max}\ge 0\iff d_{\max}\le\frac{v_{\max}^{2}}{4M\tilde{L}}
\end{equation}

\noindent Hence, by collecting our requirements \eqref{Rmax:vs:dmax}, \eqref{dmax:bound} and \eqref{quadratic:condition3} together with the fact that $d_{\max}>0$ we have
\begin{align}
\bullet\qquad & 0<d_{\max}\le\frac{v_{\max}^{2}}{4M\tilde{L}} \label{dmax:bound2} \\
\bullet\qquad & \frac{v_{\max}-\sqrt{v_{\max}^{2}-4M\tilde{L}d_{\max}}}{2M\tilde{L}}\le\delta t\le\frac{v_{\max}+\sqrt{v_{\max}^{2}-4M\tilde{L}d_{\max}}}{2M\tilde{L}} \label{deltat:bound2} \\
\bullet\qquad &\frac{1}{M+v_{\max}}d_{\max}\le\delta t \label{deltat:vs:dmax}
\end{align}

\noindent By defining
\begin{equation}\label{function:h}
h(d_{\max}):=\frac{v_{\max}-\sqrt{v_{\max}^{2}-4M\tilde{L}d_{\max}}}{2M\tilde{L}}
\end{equation}

\noindent we obtain that
\begin{equation*}
h'(d_{\max})=\frac{1}{\sqrt{v_{\max}^{2}-4M\tilde{L}d_{\max}}}
\end{equation*}

\noindent Hence, $h'(\cdot)$ is strictly increasing for $0\le d_{\max}<\frac{v_{\max}^{2}}{4M\tilde{L}}$ and furthermore
\begin{equation*}
h'(0)=\frac{1}{v_{\max}};\quad h(0)=0
\end{equation*}

\noindent The later implies that
\begin{equation}\label{function:h:property}
h(d_{\max})\ge\frac{1}{M+v_{\max}}d_{\max},\forall d_{\max}\in\left(0,\frac{v_{\max}^{2}}{4M\tilde{L}}\right]
\end{equation}

\noindent Thus it follows from  \eqref{dmax:interval}, \eqref{deltat:interval}, \eqref{function:h} and \eqref{function:h:property} that \eqref{dmax:bound2}-\eqref{deltat:vs:dmax} are fulfilled (see also Figure 1).

\begin{figure}[H]\label{fig1}
\begin{center}
\begin{tikzpicture}[
axis/.style={very thick, ->, >=stealth'}]
\begin{scope}[rotate=-90, style={very thick}]
        \draw[domain= -2:2, color=orange, bottom color=green, top color=green] plot (\x,4-\x*\x);
\end{scope}
\draw[axis] (-0.5,-2) -- (6,-2) node [right] {$d_{\max}$};
\draw[axis] (0,-2.5) -- (0,3) node [above] {$\delta t$};
\draw[dashed, very thick, color=blue] (0,-2) -- (6,-0.5)  node [right, color=black] {$\frac{1}{v_{\max}}d_{\max}$};
\draw[very thick, color=red] (0,-2) -- (6,-1) node [right, color=black] {$\frac{1}{M+v_{\max}}d_{\max}$};
\draw[dashed] (4,-2) -- (4,0);
\fill[black] (4,0) circle (2pt);
\fill[black] (4,-2) circle (2pt) node[below right] {$\frac{v_{\max}^{2}}{4M\tilde{L}}$};
\coordinate [label=left:$\frac{v_{max}}{M\tilde{L}}$] (A) at (0,2);
\end{tikzpicture}
\caption{\small \sl Feasible $d_{\max}-\delta t$ region}
\end{center}
\end{figure}
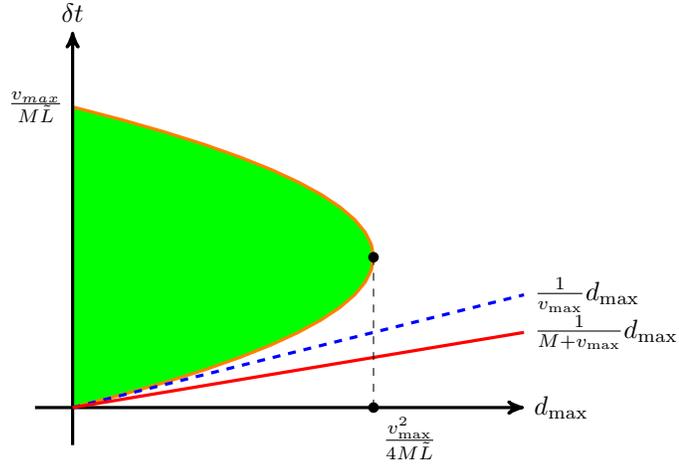

\noindent \textbf{STEP 3: Selection of cell index $l_i'$ and verification of Condition (C).}

Let $\tilde{x}_{i}(\cdot)$ be the solution of the reference trajectory as given by \eqref{reference:solution} and  $l_i'\in\mathcal{I}$ the index of a cell $S_{l_i'}$ with $\tilde{x}_i(\delta t)\in S_{l_i'}$. We prove that for any initial cell configuration $\bar{l}=(\bar{l}_{1},\ldots,\bar{l}_{N})\in\mathcal{I}^{N}$ with ${\rm pr}_i(\bar{l})=\tilde{l}_i$, selection of feedback laws in \eqref{feedback:for:others} which satisfy Property (P) for all $\hat{i}\in\{1,\ldots,N\}\setminus \{{i}\}$ and for each initial condition $x_{i0}\in S_{l_i}$ of $i$ and $x_{\hat{i}0}\in S_{\bar{l}_{\hat{i}}}$ of the other agents $\hat{i}\in\{1,\ldots,N\}\setminus \{{i}\}$, the solution of the closed loop system \eqref{full:information:dynamics}-\eqref{feedback:ki}-\eqref{feedback:for:others} is defined for all $t\in [0,\delta t]$ and the trajectory $x_i(\cdot)$ of agent $i$ at $\delta t$ satisfies

\begin{equation}\label{xi:eq:tildexi}
x_i(\delta t)=\tilde{x}_i(\delta t)
\end{equation}

\noindent namely coincides with the endpoint of the reference trajectory (see Figure 2).

\begin{figure}[H]\label{fig2}
\begin{center}
\begin{tikzpicture}

\draw[color=green,thick] (-0.3,-0.5) -- (0.3,-0.7) -- (0.8,0.1) -- (0.1,0.8)-- (-0.3,-0.5);
\draw[color=green,] (0.3,-0.7) -- (0.8,0.1) -- (1.8,0.8)-- (1.6,-0.5) -- (0.3,-0.7);
\draw[color=green,] (0.1,0.8) -- (1.8,0.8);
\draw[color=green,] (1.8,0.8)-- (2.8,-0.7)-- (1.6,-0.5);
\draw[color=cyan,thick] (1.8,0.8)-- (2.8,-0.7)-- (3.3,0.3) -- (3,1) -- (1.8,0.8);

\fill[black] (3,0.5) circle (2pt) node[above right] {$x_{i}(\delta t)$};

\draw[color=black,<->,thick] (-1.5,-0.9) -- (-1.5,1.5);

\coordinate [label=left:$\textcolor{green}{S_{l_i}}$] (A) at (-0.3,-0.5);
\coordinate [label=right:$\textcolor{cyan}{S_{l_i'}}$] (A) at (2.8,-0.7);
\coordinate [label=left:$d_{\max}$] (A) at (-1.5,0);

\draw[color=red,->,thick] (0,0) -- (3,0.5);
\draw[color=red,->] (0.2,0.6) -- (3,0.5);
\draw[color=red,->] (0.4,-0.4) -- (3,0.5);

\draw[color=black,thin] (0,0.3) circle[radius=1.2cm];
\fill[black] (0.2,0.6) circle (1.5pt) node[above left] {$x_{i0}$};
\fill[black] (0.4,-0.4) circle (1.5pt);
\fill[black] (0,0) circle (2pt) node[above left]{$\tilde{x}_{i0}$};

\end{tikzpicture}
\end{center}
\caption{}
\end{figure}

\noindent We first note that due to the result of Proposition \ref{completeness:result}, the solution of the closed loop system is defined on the whole interval $[0,\delta t]$. In order to show \eqref{xi:eq:tildexi} we show that $x_i(\cdot)$ is an appropriate modification of the reference trajectory $\tilde{x}_i(\cdot)$. In particular, it holds

\begin{equation} \label{xi:eq:tildexi:solution}
x_{i}(t)=\tilde{x}_i(t)+\left(1-\frac{t}{\delta t}\right)(x_{i0}-\tilde{x}_{i0}),\forall t\in [0,\delta t]
\end{equation}

\noindent which implies \eqref{xi:eq:tildexi}. Indeed, by adopting again the notation of Step 1, we have from \eqref{reference:solution}, \eqref{full:information:dynamics}, \eqref{feedback:ki}, \eqref{ki:Deltai} and \eqref{fi:Deltai} that
\begin{align*}
\dot{\tilde{x}}_{i}(t)&=\tilde{f}_{i,\tilde{l}_i}(\tilde{x}_{i}(t)) \\
\dot{x}_{i}(t)&=f_{i}(x(t))+k_{i,\tilde{l}_i}(t,x_i(t),\bar{x}(t);x_{i0}) \\
&=\tilde{f}_{i,\tilde{l}_i}(\tilde{x}_{i}(t))+\Delta_{i,\tilde{l}}(x_{i}(t),\bar{x}(t))+k_{i,\tilde{l}_i,1}(x_i(t),\bar{x}(t))+k_{i,\tilde{l}_i,2}(x_{i0})+k_{i,\tilde{l}_i,3}(t;x_{i0}) \\
&=\tilde{f}_{i,\tilde{l}_i}(\tilde{x}_{i}(t))+k_{i,\tilde{l}_i,2}(x_{i0})+k_{i,\tilde{l}_i,3}(t;x_{i0})
\end{align*}

\noindent and hence, we get
\begin{align*}
\tilde{x}_{i}(t)&=\tilde{x}_{i0}+\int_{0}^{t}\tilde{f}_{i,\tilde{l}_i}(\tilde{x}_{i}(s))ds \\
x_{i}(t)&=x_{i0}+\int_{0}^{t}(\tilde{f}_{i,\tilde{l}_i}(x_i(s))+k_{i,\tilde{l}_i,2}(x_{i0})+k_{i,\tilde{l}_i,3}(s;x_{i0}))ds
\end{align*}

\noindent Then it follows from \eqref{feedback:ki2} and \eqref{feedback:ki3} that
\begin{align*}
x_{i}(t)-\tilde{x}_{i}(t)&=x_{i0}-\tilde{x}_{i0}+\int_{0}^{t}[\tilde{f}_{i,\tilde{l}_i}(x_{i}(s))-\tilde{f}_{i,\tilde{l}_i}(\tilde{x}_{i}(s))+k_{i,\tilde{l}_i,2}(x_{i0})+k_{i,\tilde{l}_i,3}(s;x_{i0})]ds \\
&=x_{i0}-\tilde{x}_{i0}+\int_{0}^{t}[\tilde{f}_{i,\tilde{l}_i}(x_{i}(s))-\tilde{f}_{i,\tilde{l}_i}(\tilde{x}_{i}(s)) \\
&\left.-\tilde{f}_{i,\tilde{l}_i}\left(\tilde{x}_{i}(s)+\left(1-\frac{s}{\delta t}\right)(x_{i0}-\tilde{x}_{i0})\right)+\tilde{f}_{i,\tilde{l}}(\tilde{x}_i(s))-\frac{1}{\delta t}(x_{i0}-\tilde{x}_{i0})\right]ds \\
&=\left(1-\frac{t}{\delta t}\right)(x_{i0}-\tilde{x}_{i0}) \\
&+\int_{0}^{t}\left[\tilde{f}_{i,\tilde{l}_i}(x_i(s))-\tilde{f}_{i,\tilde{l}_i}\left(\tilde{x}_{i}(s)+\left(1-\frac{s}{\delta t}\right)(x_{i0}-\tilde{x}_{i0})\right)\right]ds,\forall t\in [0,\delta t]
\end{align*}

\noindent Hence, we get from \eqref{tildefi:Lipschitz:const} that
\begin{align*}
|x_{i}(t)-\tilde{x}_i(t)&-\left(1-\frac{t}{\delta t}\right)(x_{i0}-\tilde{x}_{i0})| \\
&\le \int_{0}^{t}L_2\left|x_{i}(s)-\tilde{x}_{i}(s)-\left(1-\frac{s}{\delta t}\right)(x_{i0}-\tilde{x}_{i0})\right|ds,\forall t\in [0,\delta t]
\end{align*}

\noindent Application of the Gronwall Lemma implies that \eqref{xi:eq:tildexi:solution} holds and hence, that $x_{i}(\delta t)=\tilde{x}_{i}(\delta t)$ as desired.
\end{proof}

\section{Conclusions}

We have provided a framework in order to extract discrete state transition systems for multi-agent systems under coupled constraints and quantified admissible space-time discretizations which allow for well posed abstractions.

We aim at extending the approach of Proposistion \ref{admissible:discretizations} in order to derive sufficient conditions which guarantee that each agent can reach at least a minimum ($>1$) number of discrete cells in time $\delta t$. Thus we can exploit the corresponding hybrid controllers and the result of Proposition \ref{discrete:transitions:result} for motion planning. Furthermore, we intend to appropriately modify our approach for the case of bounded domains in order to obtain finite transition systems.

\section{Appendix}

\noindent \textbf{Fact I.} Consider an arbitrary set $S\in\Rat{n}$ and a constant $R>0$. Then for every $x\in\partial(S+B(R))$ it holds
$$
|x-y|\ge R,\forall y\in S
$$
\begin{proof}
Indeed, suppose on the contrary that there exists $\tilde{y}\in S$ with $|x-\tilde{y}|\le R-\varepsilon$ for certain $\varepsilon>0$. Then for all $\tilde{x}\in{\rm int}(B_{x}(\varepsilon))$ we have
$$
|\tilde{x}-\tilde{y}|\le|\tilde{x}-x|+|x-\tilde{y}|<\varepsilon+R-\varepsilon=R
$$

\noindent hence $\tilde{x}\in S+B(R)$ for all $\tilde{x}\in{\rm int}(B_{x}(\varepsilon))$ which implies that $x\notin\partial(S+B(R))$ and contradicts our statement.
\end{proof}

\end{document}